\documentclass[11pt,oneside,a4paper,DIV=11]{scrartcl}   	

\usepackage[utf8]{inputenc}

\usepackage[dvipsnames]{xcolor}
\definecolor{color1}{RGB}{230,57,70}
\definecolor{color2}{RGB}{7,36,66}
\definecolor{color3}{RGB}{50,23,77}
\usepackage[colorlinks=true,pdfusetitle,linktoc,bookmarks=false]{hyperref}

\usepackage{amssymb,amsthm,amsfonts,mathtools}
\usepackage{newtxtext,newtxmath,dsfont,microtype}
\usepackage{xspace}
\usepackage{tabularx,enumerate,booktabs}

\input{braket}

\newcommand\prob\textsc
\makeatletter
\newcommand{\problemtitle}[1]{\gdef\@problemtitle{#1}}
\newcommand{\probleminput}[1]{\gdef\@probleminput{#1}}
\newcommand{\problemquestion}[1]{\gdef\@problemquestion{#1}}
\newcommand{\problempromise}[1]{\gdef\@problempromise{#1}}
\DeclareDocumentEnvironment{problem}{o o}{
\problemtitle{\IfValueT{#1}{#1}}\probleminput{}\problempromise{}\problemquestion{}
}{
  \par\addvspace{.5\baselineskip}
  \noindent
  \begin{tabularx}{\linewidth}{@{\hspace{\parindent}} l X c}
    \multicolumn{2}{@{\hspace{\parindent}}l}{\prob{\@problemtitle}} \\
    \textbf{Input:} & \@probleminput \\
\ifx\@problempromise\empty%
\else%
    \textbf{Promise:} & \@problempromise \\
\fi%
    \textbf{Question:} & \@problemquestion
  \end{tabularx}
  \par\addvspace{.5\baselineskip}
}
\makeatother

\newcommand{\defas}{\coloneqq}

\newcommand{\lham}{{\normalfont{\textsc{Local Hamiltonian}}}\xspace}
\newcommand{\tilham}{{\normalfont{\textsc{TI-Local Hamiltonian}}}\xspace}

\newcommand{\QMA}{{\normalfont{QMA}}\xspace}
\newcommand{\QMAEXP}{{\normalfont{QMA\textsubscript{EXP}}}\xspace}
\newcommand{\PreciseQMA}{{\normalfont{PreciseQMA}}\xspace}
\newcommand{\PSPACE}{{\normalfont{PSPACE}}\xspace}
\newcommand{\YES}{{\normalfont{YES}}\xspace}
\newcommand{\NO}{{\normalfont{NO}}\xspace}
\DeclareMathOperator{\lmin}{\lambda_\mathrm{min}}

\newcommand{\ii}{\mathrm{i}}
\newcommand{\ee}{\mathrm{e}}

\newcommand{\1}{\mathds 1}

\usepackage[capitalise]{cleveref}
\Crefname{lemma}{Lemma}{Lemmas}
\Crefname{proposition}{Proposition}{Propositions}
\Crefname{definition}{Definition}{Definitions}
\Crefname{theorem}{Theorem}{Theorems}
\Crefname{conjecture}{Conjecture}{Conjectures}
\Crefname{corollary}{Corollary}{Corollaries}
\Crefname{example}{Example}{Examples}
\Crefname{section}{Section}{Sections}
\Crefname{appendix}{Appendix}{Appendices}
\Crefname{figure}{Fig.}{Figs.}
\Crefname{equation}{Eq.}{Eqs.}
\Crefname{table}{Table}{Tables}
\Crefname{item}{Property}{Properties}
\Crefname{remark}{Remark}{Remarks}

\newtheorem{theorem}{Theorem}[section]
\newtheorem{lemma}[theorem]{Lemma}

\newtheorem{corollary}[theorem]{Corollary}
\newtheorem{definition}[theorem]{Definition}

\newtheorem{remark}[theorem]{Remark}

\newcommand{\identity}{\mathds{1}}

\newcommand{\Htarget}{H_{\mathrm{target}}}
\newcommand{\HSF}{H_{\mathrm{SF}}}
\newcommand{\Huniv}{H_{\mathrm{univ}}}
\newcommand{\Hbulk}{H_{\mathrm{bulk}}}
\newcommand{\Hboundary}{H_{\mathrm{boundary}}}
\newcommand{\TPE}{T_\mathrm{PE}}
\newcommand{\MPE}{M_\mathrm{PE}}

\newcommand{\spec}{\mathrm{spec}}
\newcommand{\hyper}{\field{H}^2}

\newcommand{\field}{\mathds}

\DeclareMathOperator{\BigO}{O}
\DeclareMathOperator{\poly}{poly}
\DeclareMathOperator{\Herm}{Herm}
\DeclareMathOperator{\rank}{rank}
\DeclareMathOperator{\spn}{span}

\renewcommand{\paragraph}[1]{\textbf{\textit{#1}}}

\usepackage{todonotes}

\usepackage[backend=bibtex,style=alphabetic,doi=false,isbn=false,url=false,maxbibnames=20,sorting=none]{biblatex}
\renewbibmacro{in:}{}
\bibliography{ref-ti-quantum}
\bibliography{ref-ti-quantum-jo}

\newbibmacro{string+doi}[1]{\iffieldundef{doi}{#1}{\href{https://dx.doi.org/\thefield{doi}}{#1}}}
\DeclareFieldFormat{title}{\usebibmacro{string+doi}{\mkbibemph{#1}}}
\DeclareFieldFormat[article,thesis,incollection,inproceedings]{title}{\usebibmacro{string+doi}{\mkbibquote{#1}}}

\usepackage{authblk}

\title{Translationally-Invariant\\Universal Quantum Hamiltonians in 1D}
\author[a]{Tamara Kohler}
\author[b,c]{Stephen Piddock}
\author[d]{Johannes Bausch}
\author[a]{Toby Cubitt}
\affil[a]{Department of Computer Science, University College London, UK}
\affil[b]{School of Mathematics, University of Bristol, UK}
\affil[c]{Heilbronn Insitute for Mathematical Research, Bristol, UK}
\affil[d]{Department of Applied Mathematics and Theoretical Physics, University of Cambridge, UK}

\begin{document}

\date{}
\maketitle

\abstract{Recent work has characterised rigorously what it means for one quantum system to simulate another, and demonstrated the existence of universal Hamiltonians---simple spin lattice Hamiltonians that can replicate the entire physics of any other quantum many body system.
Previous universality results have required proofs involving complicated `chains' of perturbative `gadgets'.
In this paper, we derive a significantly simpler and more powerful method of proving universality of Hamiltonians, directly leveraging the ability to encode quantum computation into ground states.
This provides new insight into the origins of universal models, and suggests a deep connection between universality and complexity.
We apply this new approach to show that there are universal models even in translationally invariant spin chains in 1D.
This gives as a corollary a new Hamiltonian complexity result, that the local Hamiltonian problem for translationally-invariant spin chains in one dimension with an exponentially-small promise gap is PSPACE-complete.
Finally, we use these new universal models to construct the first known toy model of 2D--1D holographic duality between local Hamiltonians.
}

\clearpage
\tableofcontents


\section{Introduction}

Analog Hamiltonian simulation is one of the most promising applications of quantum computing in the NISQ (noisy, intermediate scale, quantum) era, because it does not require fully fault-tolerant quantum operations.
Its potential applications have led to an interest in constructing a rigorous theoretical framework to describe Hamiltonian simulation.

Recent work has precisely defined what it means for one quantum system to simulate another~\cite{Cubitt:2017}, and demonstrated that---within very demanding definitions of what it means for one system to simulate another---there exist families of Hamiltonians that are universal, in the sense that they can simulate all other quantum Hamiltonians.
This work was recently extended, with the first construction of a translationally invariant universal family of Hamiltonians \cite{PiddockBausch}.

Previous universality results have relied heavily on using perturbation gadgets, and constructing complicated `chains' of simulations to prove that simple models are indeed universal.
In this paper we present a new, simplified method for proving universality.
This method makes use of another technique from Hamiltonian complexity theory: history state Hamiltonians \cite{Kitaev2002}.
Leveraging the fact that it is possible to encode computation into the ground state of local Hamiltonians, we show that it is possible to prove universality by constructing Hamiltonian models which can compute the energy levels of arbitrary target Hamiltonians.

In order to ensure that the universality constructions preserve the entire physics of the target system (and not just the energy levels), we make use of an idea originally from \cite{universality_of_adiabatic_computation} and used recently in \cite{zhou:18,nirkhe:18,Aharonov_2014}: `idling to enhance coherence'.
Before computing the energy levels of the target system, the computation encoded in the simulator system `idles' in its initial state for time $L$.
By choosing $L$ to be sufficiently large, we can ensure that with high probability there is a fixed set of spins in the simulator system which map directly to the state of the target system.

As well as providing a route to simplifying previous proofs, this `history-state simulation method'  also offers more insight into the origins of universality, and the relationship between universality and complexity.
The classification of two-qubit interactions by their simulation ability in \cite{Cubitt:2017}, which showed that the universal class was precisely the set of QMA-complete interactions, was suggestive of a connection between simulation and complexity.
And a complexity theoretic classification of universal models already exists in the classical case \cite{Cubitt:2016}.
But until now it was not clear whether a connection existed for general quantum interactions, or whether it was merely an accident in the two-qubit case.
Previous methods for proving universality in the quantum case didn't offer a route to classifying universal models, and the more complicated non-commutative structure of quantum Hamiltonians meant that the techniques from the classical proof couldn't be applied.
By demonstrating that it is possible to prove universality by leveraging the ability to encode computation into ground states, we have provided a route to showing that the connection between universality and complexity holds more generally.
In a companion paper \cite{kohler2021general} we make this insight rigorous, by deriving a full complexity theoretic classification of universal quantum Hamiltonians.

We also use the `history-state simulation method'  to provide a simple construction of two new universal models.
Both of these are translationally invariant systems in 1D, and we show that one of these constructions is efficient in terms of the number of spins in the universal construction (yet not in terms of the simulating system's norm):

\begin{theorem} \label{main-intro}
There exists a  two-body interaction $h^{(1)}$ depending on a single parameter $h^{(1)}=h^{(1)}(\phi)$, and a fixed one-body interaction $h^{(2)}$ such that the family of translationally-invariant Hamiltonians on a chain of length $N$,
\begin{equation}
\Huniv(\phi, \Delta,T) = \Delta \sum_{\langle i,j \rangle} h^{(1)}_{i,j}(\phi) + T \sum_{i = 0}^N h^{(2)}_{i},
\end{equation}
is a universal model, where $\Delta$, $T$ and $\phi$ are parameters of the Hamiltonian, and the first sum is over adjacent sites along the chain.
The universal model is efficient in terms of the number of spins in the simulator system.
\end{theorem}
By tuning $\phi$, $T$ and $\Delta$, this model can replicate (in the precise sense of~\cite{Cubitt:2017}) all quantum many body physics.

This is the first translationally invariant universal model which is efficient in terms of system size overhead.
Its existence implies that, for problems which preserve hardness under simulation, complexity theoretic results for general Hamiltonians can also apply to 1D, translationally invariant Hamiltonians (though care must be taken when applying this, as the construction is not efficient in the norm of the simulating system).
This is for instance the case for a reduction from a PreciseQMA-hard local Hamiltonian (LH) problem, for which the reduction to a translationally-invariant version preserves the correct promise gap scaling.
This in turn implies that the local Hamiltonian problem remains PSPACE-hard for a promise gap that closes exponentially quickly, even when enforcing translational invariance for the couplings.
This stands in contrast to a promise gap which closes as $1/\poly$ in the system size, in which case the variant is either \QMA (for non-translational invariance) or \QMAEXP (for translational invariance) complete.

Furthermore, \cref{main-intro} allows us to construct the first toy model of holographic duality between local Hamiltonians from a 2D bulk to a 1D boundary, extending earlier work on toy models of holographic duality in~\cite{Pastawski:2015} and~\cite{HQECC-local}.

We also construct a universal model which is described by just two free parameters, but where the model is no longer efficient in the system size overhead:
\begin{theorem}
There exists a fixed two-body interaction $h^{(3)}$ and a fixed one-body interaction $h^{(2)}$ such that the family of translationally-invariant Hamiltonians on a chain of length $N$,
\begin{equation}
\Huniv(\Delta,T) = \Delta \sum_{\langle i,j \rangle} h^{(3)}_{i,j} + T \sum_{i = 0}^N h^{(2)}_{i},
\end{equation}
is a universal model, where $\Delta$ and $T$ are parameters of the Hamiltonian, and the first sum is over adjacent sites along the chain.
\end{theorem}
By varying the size of the chain $N$ that this Hamiltonian is acting on, and tuning the $\Delta$ and $T$ parameters in the construction, this Hamiltonian can replicate (again in the precise sense of~\cite{Cubitt:2017}) all quantum many body physics.
We are able to demonstrate that constructing a universal model with no free parameters is not possible, but the existence of a universal model with just one free parameter is left as an open question.

The remainder of the paper is set out as follows.
In \cref{preliminaries} we cover the necessary background regarding the theory of simulation, and encoding computation into ground states of \QMA-hard Hamiltonians.
In \cref{sec:overview} we give an overview of the new method for proving universality, and our two new universal constructions. 
Reading these sections should be enough to gain an intuitive understanding of our approach and our results.
The full proofs of our results are given in \cref{universality} - this section may be skipped on an initial reading if you are primarily interested in understanding the general approach, or the applications of the results.
The complexity theory implications are discussed in  \cref{sec:complexity}, while in \cref{holography-implications} the new toy model of holographic duality is constructed.
Avenues for future research, are discussed in \cref{discussion}.

\section{Preliminaries} \label{preliminaries}

\subsection{Universal Hamiltonians}
\subsubsection{Hamiltonian Encodings}
Any simulation of a Hamiltonian $H$ by another Hamiltonian $H'$ must involve ``encoding'' $H$ in $H'$ in some fashion.
In~\cite{Cubitt:2017} it was shown that any encoding map $\mathcal{E}(A)$ which satisfies three basic requirements
\begin{enumerate}[i)]
\item $\mathcal{E}(A) = \mathcal{E}(A)^\dagger$ for all $A \in \text{Herm}_n$
\item $\spec(\mathcal{E}(A)) = \spec(A)$ for all $A \in \text{Herm}_n$
\item $\mathcal{E}(pA + (1-p)B) = p \mathcal{E}(A) + (1-p)\mathcal{E}(B)$ for all $A, B \in \text{Herm}_n$ and all $p \in [0,1]$
\end{enumerate}
must be of the form
\begin{equation} \label{encodings}
\mathcal{E}(A) = V \left(A \otimes P + \overline{A} \otimes Q \right) V^\dagger,
\end{equation}
where $V$ is an isometry, $\overline{A}$ denotes complex conjugation, and $P$ and $Q$ are orthogonal projectors.
Moreover, it is shown that, under any encoding of the form given in \cref{encodings}, $\mathcal{E}(H)$ will also preserve the measurement outcomes, time evolution and partition function of $H$.

A \emph{local} encoding is an encoding which maps local observables to local observables, defined as follows.
\begin{definition}[Local subspace encoding (Definition 13 from~\cite{Cubitt:2017})]\label{def:local-encoding}
Let
\[
\mathcal{E}: \mathcal{B}\left(\otimes_{j=1}^n \mathcal{H}_j \right) \rightarrow \mathcal{B}\left(\otimes_{j=1}^{n} \mathcal{H}'_j \right)
\]
 be a subspace encoding. We say that the encoding is local if for any operator $A_j \in \Herm(\mathcal{H}_j)$ there exists $A'_j \in \Herm(\mathcal{H}'_j)$ such that:
\[
\mathcal{E}(A_j \otimes \identity) = (A'_j \otimes \identity)\mathcal{E}(\identity).
\]
\end{definition}

It is shown in~\cite{Cubitt:2017} that if an encoding $\mathcal{E}(M) = V(M \otimes P + \overline{M} \otimes Q){V}^\dagger$ is local, then the isometry $V$ can be decomposed into a tensor product of isometries $V = \otimes_i V_i$,  for isometries $V_i: \mathcal{H}_i \otimes E_i \rightarrow \mathcal{H}'_i$, for some ancilla system $E_i$.\footnote{There is a more general definition of simulation which doesn't require the isometries to be tensor product \cite{BH17}. However, these types of simulations don't preserve the local structure of the Hamiltonian. So while they are interesting from a complexity theoretic perspective, they are not as useful for physical simulation.}

In this paper all of the encodings we work with are of the simpler form $\mathcal{E}(A) = VAV^\dagger$.

\subsubsection{Hamiltonian Simulation}
Building on encodings, \cite{Cubitt:2017} developed a rigorous formalism of Hamiltonian simulation, formalizing the notion of one many-body system reproducing identical physics as another system, including the case of approximate simulation and simulations within a subspace.
We first describe the simpler special case of \emph{perfect} simulation.
If $H'$ perfectly simulates $H$, then it \emph{exactly} reproduces the physics of $H$ below some energy cutoff $\Delta$, where $\Delta$ can be chosen arbitrarily large.
For brevity, we abbreviate the low-energy subspace of an operator $A$ via $S_{\le\Delta(A)} \defas \spn\{ \ket{\psi} : A \ket{\psi} = \lambda\ket\psi \land \lambda \le \Delta \}$.
\begin{definition}[{Exact simulation,~\cite[Def.~20]{Cubitt:2017}}]\label{def:exact-sim}
We say that $H'$ perfectly simulates $H$ below the cutoff energy $\Delta$ if there is a local encoding $\mathcal{E}$ into the subspace $S_{\mathcal{E}}$ such that
  \begin{enumerate}[i.]
  \item $S_{\mathcal{E}} = S_{\leq\Delta(H')}$, and
  \item $H'|_{\leq \Delta} = \mathcal{E}(H)|_{S_\mathcal{E}}$.
  \end{enumerate}
\end{definition}

We can also consider the case where the simulation is only approximate:
\begin{definition}[{Approximate simulation,~\cite[Def.~23]{Cubitt:2017}}]\label{app-sim}
Let $\Delta,\eta,\epsilon>0$.
A Hamiltonian $H'$ is a $(\Delta, \eta, \epsilon)$-simulation of the Hamiltonian $H$ if there exists a local encoding $\mathcal{E}(M)=V(M \otimes P + \overline{M} \otimes Q){V}^\dagger$ such that
\begin{enumerate}[i.]
\item There exists an encoding $\tilde{\mathcal{E}}(M)=\tilde{V}(M \otimes P + \overline{M} \otimes Q)\tilde{V}^\dagger$ into the subspace $S_{\tilde{\mathcal E}}$ such that $S_{\tilde{\mathcal{E}}} = S_{\leq \Delta(H')}$ and $\|\tilde{V} - V\| \leq \eta$; and
\item $\|H'_{\leq \Delta} - \tilde{\mathcal{E}}(H)\| \leq \epsilon$.
\end{enumerate}
\end{definition}
\noindent Note that the role of $\tilde{\mathcal E}$ is to provide an \emph{exact} simulation as per \cref{def:exact-sim}. However, it might not always be possible to construct this encoding in a local fashion. The local encoding $\mathcal E$ in turn approximates $\tilde{\mathcal E}$, such that the subspaces mapped to by the two encodings deviate by at most $\eta$. $\epsilon$ controls how much the eigenvalues are allowed to differ.

If we are interested in whether an infinite family of Hamiltonians can be simulated by another, the notion of overhead becomes interesting: if the system size grows, how large is the overhead necessary for the simulation, in terms of the number of qudits, operator norm or computational resources?
We capture this notion in the following definition.
\begin{definition}[{Simulation,~\cite[Def.~23]{Cubitt:2017}}]\label{def:efficient-sim}
We say that a family $\mathcal{F}'$ of Hamiltonians can simulate a family $\mathcal{F}$ of Hamiltonians if, for any $H \in \mathcal{F}$ and any $\eta, \epsilon > 0$ and $\Delta \geq \Delta_0$ (for some $\Delta_0 > 0$), there exists $H' \in \mathcal{F}'$ such that $H'$ is a $(\Delta, \eta, \epsilon)$-simulation of $H$.

We say that the simulation is efficient if, in addition, for $H$ acting on $n$ qudits and $H'$ acting on $m$ qudits, $\|H'\| = \poly(n, 1 / \eta,1 / \epsilon,\Delta)$ and $m = \poly(n, 1 / \eta,1 / \epsilon,\Delta)$; $H'$ is efficiently computable given $H$, $\Delta$, $\eta$ and $\epsilon$; each local isometry $V_i$ in the decomposition of $V$ is itself a tensor product of isometries which map to $\BigO(1)$ qudits; and there is an efficiently constructable state $\ket{\psi}$ such that $P \ket{\psi} = \ket{\psi}$.
\end{definition}

As already outlined, in~\cite{Cubitt:2017} it is shown that approximate Hamiltonian simulation preserves important physical properties. We recollect the most important ones in the following.
\begin{lemma}[{\cite[Lem.~27, Prop.~28, Prop.~29]{Cubitt:2017}}] \label{physical-properties}
Let $H$ act on $(\field{C}^d)^{\otimes n}$.
  Let $H'$ act on $(\field{C}^{d'})^{\otimes m}$, such that $H'$ is a $(\Delta, \eta, \epsilon)$-simulation of $H$ with corresponding local encoding $\mathcal{E}(M) = V(M \otimes P + \overline{M} \otimes Q)V^\dagger$.
  Let $p = \rank(P)$ and $q = \rank(Q)$.
  Then the following holds true.
  \begin{enumerate}[i.]
  \item Denoting with $\lambda_i(H)$ (resp.\ $\lambda_i(H')$) the $i$\textsuperscript{th}-smallest eigenvalue of $H$ (resp.\ $H'$), then for all $1 \leq i \leq d^n$, and all $(i-1)(p+q) \leq j \leq i (p+q)$, $|\lambda_i(H) - \lambda_j(H')| \leq \epsilon$.
  \item The relative error in the partition function evaluated at $\beta$ satisfies
    \begin{equation}
      \frac{|\mathcal{Z}_{H'}(\beta) - (p+q)\mathcal{Z}_H(\beta) |}{(p+q)\mathcal{Z}_H(\beta)} \leq \frac{(d')^m \ee^{-\beta \Delta}}{(p+q)d^n \ee^{-\beta \|H\|}} + (\ee^{\epsilon \beta} - 1).
    \end{equation}
  \item For any density matrix $\rho'$ in the encoded subspace for which $\mathcal{E}(\identity)\rho' = \rho'$, we have
    \begin{equation}
      \|\ee^{-\ii H't}\rho'\ee^{\ii H't} - \ee^{-\ii \mathcal{E}(H)t}\rho'\ee^{\ii \mathcal{E}(H)t}\|_1 \leq 2\epsilon t + 4\eta.
    \end{equation}
  \end{enumerate}
\end{lemma}

\Cref{def:efficient-sim} naturally leads to the question in which cases a family of Hamiltonians is \emph{so} versatile that it can simulate any other Hamiltonian: in that case, we call the family \emph{universal}.
\begin{definition}[{Universal Hamiltonians~\cite[Def.~26]{Cubitt:2017}}]
We say that a family of Hamiltonians is a universal simulator---or simply is universal---if any (finite-dimensional) Hamiltonian can be simulated by a Hamiltonian from the family.
We say that the universal simulator is efficient if the simulation is efficient for all local Hamiltonians.
\end{definition}

\subsection{Circuit-to-Hamiltonian Mappings}\label{sec:c-to-ham}
The key idea behind our universal constructions is that it is possible to encode computation into the ground state of local Hamiltonians.
This technique was first proposed by Feynman in~\citeyear{Feynman1986}, and is the foundation for many prominent results in Hamiltonian complexity theory, such as \QMA-hardness of the local Hamiltonian problem~\cite{Feynman1986,Kitaev2002}. 

For the constructions we develop in this paper, we will make use of the ability to encode an arbitrary quantum computation into the ground state of a local Hamiltonian. These are often called ``circuit-to-Hamiltonian mappings'', though the mappings may involve other models of quantum computation than the circuit model.
These Hamiltonians are typically constructed in such a way that their ground states are ``computational history states''. A very general definition of history states was given in~\cite{CG18}; we will only require the simpler ``standard'' history states here.
\begin{definition}[Computational history state]\label{def:history-state}
A computational history state $\ket{\Phi}_{CQ} \in \mathcal{H}_C \otimes \mathcal{H}_Q$ is a state of the form
\[ \
\ket{\Phi}_{CQ} = \frac{1}{\sqrt{T}} \sum_{t=1}^{T} \ket{\psi_t}\ket{t},
\]
where $\{\ket{t}\}$ is an orthonormal basis for $\mathcal{H}_C$ and $\ket{\psi_t} = \Pi_{i=1}^tU_i\ket{\psi_0}$ for some initial state $\ket{\psi_0}\in \mathcal{H}_Q$ and set of unitaries $U_i \in \mathcal{B}(\mathcal{H}_Q)$.

$\mathcal{H}_C$ is called the clock register and $\mathcal{H}_Q$ is called the computational register.
If $U_t$ is the unitary transformation corresponding to the $t$\textsuperscript{th} step of a quantum computation, then $\ket{\psi_t}$ is the state of the computation after $t$ steps.
We say that the history state $\ket{\Phi}_{CQ}$ encodes the evolution of the quantum computation.
\end{definition}
Note that $U_t$ need not necessarily be a gate in the quantum circuit model. It could also e.g.\ be one time-step of a quantum Turing machine, or even a time-step in some more exotic model of quantum computation~\cite{Bausch2016}, or an isometry~\cite{Usher2017}.
In the particular constructions we make use of in this work, $U_t$ will be a time-step of a quantum Turing machine.

\section{Overview of construction} \label{sec:overview}

\subsection{High-level outline of the construction} \label{sec:outline}

As mentioned in \cref{sec:c-to-ham}, the key technique we make use of in our universality constructions is the ability to encode computations into the ground states of local Hamiltonians.
The model of computation we encode is the quantum Turing machine (QTM) model - standard techniques for encoding QTMs in local Hamiltonians give translationally invariant Hamiltonians \cite{spec-gap,Gottesman2009}.

In both the constructions we develop in this work a description of the Hamiltonian to be simulated (the ``target'' Hamiltonian, $\Htarget$) is encoded in the binary expansion of some natural number, $x \in \field N$ .
Details of this encoding are given in \cref{sec:digital-encoding}.
The natural number $x$ is then itself encoded in some parameter of the universal Hamiltonian (see \cref{sec:ti-c-to-ham} for two methods of encoding natural numbers in parameters of universal Hamiltonians). 

The Hamiltonian we use to construct the universal model has as its ground state computational history states (cf \cref{def:history-state}) which encode two QTMs ($M_1$ and $\MPE$) which share a work tape.
The two computations are `dovetailed' together - the computation $M_1$ occurs first, and the result of this computation is used as input for $\MPE$.
The first QTM, $M_1$, extracts the binary expansion of $x$ from the parameter of the Hamiltonian.
At the end of $M_1$'s computation, the binary expansion of $x$ is written on the work tape which $M_1$ shares with $\MPE$.
An outline of the methods we use to extract $x$ and write it on the Turing machine tape are given in \cref{sec:ti-c-to-ham}.

The second QTM, $\MPE$ reads in $x$, which contains a description of $\Htarget$, from the work tape which it shares with $M_1$.
It also reads in an input state $\ket{\psi}$ - this is unconstrained by the computation (it can be thought of as carrying out the same role as a witness in a \QMA verification circuit).
It then carries out phase estimation on $\ket{\psi}$ with respect to the unitary generated by $\Htarget$.

The Hamiltonian which encodes $M_1$ and $\MPE$ has a zero-energy degenerate ground space, spanned by history states with all possible input states $\ket{\psi}$.
In order to recreate the spectrum of $\Htarget$ we need to break this degeneracy.
We achieve this by adding one body projectors to the universal Hamiltonian which give the correct energy to the output of $\MPE$ to reconstruct the spectrum of $\Htarget$.

With this construction the energy levels of the universal Hamiltonian recreate the energy levels of $\Htarget$.
To ensure that the eigenstates are also correctly simulated, before $M_1$ carries out its computation, it `idles' in its initial state for some time $L$.
By choosing $L$ large enough, we show that this construction can approximately simulate any target Hamiltonian.
A more detailed sketch of how we use idling and phase estimation to achieve simulation is given in \cref{sec:dovetailing}, while rigorous proofs are given in \cref{universality}.

\subsection{A Digital Representation of a Local Hamiltonian}\label{sec:digital-encoding}
\newcommand{\ivar}{\eta}
As discussed in \cref{sec:outline}, we need to encode a description of the target Hamiltonian $\Htarget$ in some parameter of the universal Hamiltonian.
In \cref{sec:ti-c-to-ham} we outline the two methods we use to encode a natural number in the parameter of a Hamiltonian.
But how do we represent $\Htarget=\sum_{i=1}^m h_i$ in the binary expansion of a natural number $x\in\field N$, irrespective of its origin?

We will assume that $\Htarget$ is a $k$-local Hamiltonian, acting on $n$ spins of local dimension $d$.
We emphasize that $k$ can be taken to be $n$, i.e.\ the system size---and therefore we can simulate \emph{any} Hamiltonian, not just local ones. 
However we keep track of the locality parameter $k$ as it is relevant when deriving the overhead of our simulations.

Every value needed to specify the $k$-local simulated system $\Htarget$ will be represented in Elias-$\gamma'$ coding, which is a simple self-delimiting binary code which can encode all natural numbers~\cite{Fenwick:2003,kohler:18}.
For the purpose of the encoding, we will label the $n$ spins in the system to be simulated by integers $i = 1,\ldots,n$.

The encoding of $\Htarget$ begins with the three meta-parameters $n$ (spin count), followed by $k$ (locality), and then $m$ (number of $k$-local terms).
Each of the $m$ $k$-local terms in $H$ is then specified by giving the label of the spins involved in that interaction, followed by a description of each term of the $d^k \times d^k$ Hermitian matrix describing that interaction.
Each such matrix entry is specified by giving two integers $a$ and $b$. The matrix entry can be recovered by calculating $a \sqrt{2} - b$, which is accurate up to a small error.%
\footnote{Note that by Weyl's equidistribution theory $\sqrt{2}a \mod 1$ uniformly covers $[0,1]$; the set $\mathcal{T} = \{a\sqrt{2}-b\mid a, b \in \field{Z}^+\}$ is dense in $\field{R}$.
}

Specifying $\Htarget$ to accuracy $\delta$ requires each such matrix entry to be specified to accuracy $\delta/(md^{2k})$. Therefore the length of the description of $\Htarget$ is
\begin{equation}
md^{2k}\log\left(\| \Htarget \| md^{2k}/\delta\right) =\poly\left(n,d^k,\log(\|H\|/\delta) \right)
\end{equation}

Finally, the remaining digits of $x$ specify $\Xi$---the bit precision to with which the phase estimation algorithm should calculate the energies (i.e.\ we require QPE to extract $\Xi$ binary digits), and $L$---the length of time the system should ``idle'' in its initial state before beginning its computation.

So, the binary expansion $B(x)$ of $x$ has the following form:
\begin{equation}\label{eq:B(x)}
B(x) \defas \gamma'(n) \cdot \gamma'(k) \cdot \gamma'(m) \cdot \left[ \gamma'(i)^{\cdot k} \cdot \left(\gamma'(a_j) \cdot \gamma'(b_j) \right)^{4^k}\right]^{\cdot m} \cdot \gamma'(\Xi) \cdot \gamma'(L).
\end{equation}
Here $\gamma'(n)$ denotes $n$ in Elias-$\gamma'$ coding, and $\cdot$ denotes concatenation of bit strings.

With regards to the identification of a real number $n=\sqrt 2 a - b$, we observe that it is clearly straightforward to recover $n$ from $a$ and $b$ (by performing basic arithmetic). The other direction works as follows.
\begin{remark}\label{rem:2a-b}
Let $n\in\field N$, and let $\Xi\in\field N$ denote a precision parameter. Then we can find numbers $a,b\in\field N$ such that
\[
    \left| n - \sqrt2a + b\right| \le 2^{-\Xi},
\]
and the algorithm runs in $\BigO(\poly(\Xi, \log_2 n))$.
\end{remark}
\begin{proof}
We solve $2^\Xi n = \lfloor 2^\Xi\sqrt 2 \rfloor a - 2^\Xi b$ as a linear Diophantine equation in the variables $a$ and $b$, with largest coefficient $\BigO(2^\Xi n)$.
This can be done in polynomial time in the bit precision of the largest coefficient, for instance by using the extended Euclidean algorithm~\cite{Fox2000}.
\end{proof}

In \cref{universality}, we describe a construction to $(\Delta',\eta,\epsilon')$-simulate the Hamiltonian described by $x$, but note that this will only give a $(\Delta',\eta,\epsilon'+\delta)$-simulation of the actual target Hamiltonian $\Htarget$.

\subsection{Encoding the target Hamiltonian in parameters of the simulator Hamiltonian} \label{sec:ti-c-to-ham}

In \cref{sec:digital-encoding} we described how we encode the information about the Hamiltonian we want to simulate, $\Htarget$ in a natural number $x$. 
Now we require a method to encode $x$ in some parameter of the universal Hamiltonian, and a method to write its binary expansion on the Turing machine tape shared by $M_1$ and $\MPE$.
We develop two constructions, building on the mappings in~\cite{spec-gap} and~\cite{Gottesman2009}.
The first construction is efficient in terms of the number of spins in the simulator system, while the second construction is not efficient, but requires less parameters to specify the universal model.
In both cases the computation encoded in the ground state of the Hamiltonian is a QTM, and the mapping from a QTM to the Hamiltonian gives a translationally invariant Hamiltonian.

\subsubsection{Encoding the target Hamiltonian in a phase of the simulator Hamiltonian}\label{sec:spec-gap}

First we consider the construction building on the work in \cite{spec-gap}.
Here, we encode the natural number $x \in \field N$ in a phase $\phi = x/2^{\lceil \log_2 x \rceil}$ of the Hamiltonian.

The Hamiltonian for this construction is given by $H = \sum_{i=1}^N h^{(i,i+1)}$ where $N$ is the number of spins in the simulator system, and $h$ is a two-body interaction of the form \cite[Theorem 32]{spec-gap-full}:
\begin{equation}
h = A + (e^{i\pi \phi} B + e^{i \pi 2^{-|\phi|}}C + \mathrm{h.c.})
\end{equation}
where $A$ is a fixed Hermitian matrix and $B,C$ are fixed non-Hermitian matrices.
For a detailed construction of the terms in the Hamiltonian we refer the interested reader to \cite[Section 4]{spec-gap-full}. 

The circuit-to-Hamiltonian map encodes two Turing machine computations ``dovetailed'' together, where the two Turing machines share a work tape.
The first computation is a phase estimation algorithm.
It extracts the phase $\phi$ from the Hamiltonian, and writes its binary expansion onto the work tape.  
The second computation will be outlined in \cref{sec:dovetailing}.

In order to extract $a$ digits from a phase $\phi=0.\phi_1\phi_2\cdots\phi_a\phi_{a+1}\cdots$, we require a runtime of $2^a$.
In our case, we have $a = |x| = \poly\left(n,d^k,\log(\|H\|/\delta) \right)$, where $|x|$ denotes the number of digits in the binary expansion of $x$.
As our computation is encoded as a computational history state, this in turn means that the spectral gap of the history state Hamiltonian necessarily closes as $\BigO(2^{-\poly\left(n,d^k,\log(\|H\|/\delta) \right)})$~\cite{Bausch2016a,crosson-bowen,CG18}.
This scaling of the spectral gap means that the universal model constructed via this method is not efficient in terms of the norm of the simulator system (see \cref{main-theorem-1} for full discussion of the scaling).

However, it is important to note that using the construction from \cite{spec-gap} it is possible to encode a computation with exponential runtime into a Hamiltonian on polynomially many spins.
Details of the construction are given in \cite[Section 4.5]{spec-gap-full} (in particular the relevant scaling is discussed on \cite[Page 81]{spec-gap-full}).
We will not give the details of the construction here, but note that it encodes a Turing machine which runs for $\BigO(N\exp(N))$ time steps in a Hamiltonian acting on $N$ spins \cite[Proposition 45]{spec-gap-full}.
Therefore, the universal model constructed via this method is efficient in terms of the number of spins in the simulator system.

\subsubsection{Encoding the target Hamiltonian in the size of the simulator system}

Our second construction builds on the mapping in \cite{Gottesman2009}.
Here, we encode the description of the $\Htarget$ into the binary expansion of $N$ - the number of spins the universal Hamiltonian is acting on.
 
The circuit-to-Hamiltonian map encodes two Turing machine computations ``dovetailed'' together, where again the two Turing machines share a work tape.
The first Turing machine is a binary counter Turing machine.
After it has finished running, the binary expansion of $N$ is written on the Turing machine's work tape.
In our construction, the binary expansion of $N$ contains the description of $\Htarget$.
We will discuss the second computation in \cref{sec:dovetailing}.

The binary counter QTM takes time $N$ to write out the binary expansion of $N$ on its work tape.
Since $\Htarget$ is encoded in the binary expansion of $N$, this run time, as well as the size of the simulator system is exponential in the size of the target system. 
Moreover, since the runtime is exponential in the size of the target system, the spectral gap of the universal Hamiltonian closes exponentially fast.
Therefore, the universal model constructed via this method is not efficient in terms of number of spins or the norm of the simulator system.
See \cref{main-theorem-2} for a full discussion of the scaling of this universal model.

In this case the interactions of the Hamiltonian are entirely fixed - they enforce that the ground state of the Hamiltonian is a history state encoding a QTM computation (for a detailed construction of the terms in the Hamiltonian we refer readers to \cite{Gottesman2009}.
There are two additional global parameters in the Hamiltonian which depend on the accuracy of the simulation - we defer discussion of those parameters to the technical proofs of \cref{lem:dovetailed-qpe} and \cref{main-theorem-2}.
All the information about the target Hamiltonian (the Hamiltonian to be simulated) is entirely encoded in the binary expansion of $N$ - the number of spins in the simulator system.

\subsection{Dovetailing for simulation}\label{sec:dovetailing}

After the computation carried out by $M_1$ has finished, the binary expansion of $x$ is written out on the work-tape shared by $M_1$ and $\MPE$. 
We then construct (using standard techniques from~\cite{spec-gap,Gottesman2009}) a Hamiltonian such that the two Turing machines $M_1$ and $\MPE$ share a work tape.
At the beginning of its computation, $\MPE$ reads in a description of the target Hamiltonian $H$ that we wish to simulate.
$\MPE$ then carries out phase estimation on some input state $\ket{\psi}$ (left unconstrained, just like a \QMA witness)\footnote{Although quantum phase estimation takes as input an eigenvector of the unitary, we show in the proof that this suffices, as the argument then extends to general input states by linearity.} with respect to the unitary generated by the target Hamiltonian, $U = \ee^{\ii H\tau}$ for some $\tau$ such that $\| H\tau\| < 2 \pi$.
It then outputs the eigenphase $\phi$ in terms of a pair of natural numbers $(a,b)$ such that $\phi=a\sqrt{2}-b$ (which can be done efficiently via \cref{rem:2a-b}).

The ground space of the Hamiltonian which encodes the computation of $M_1$ and $\MPE$ has zero energy, and is spanned by history states in a superposition over all possible initial states $\ket\psi$.
In general the Hamiltonian we want to simulate doesn't have a highly degenerate zero energy ground state, so we need to break this degeneracy and construct the correct spectrum for $\Htarget$.
In order to break the degeneracy and reconstruct the spectrum of $\Htarget$, we add one body projectors to the universal Hamiltonian, which are tailored such that the QPE output $(a,b)$ identifies the correct energy penalty to inflict.

In order to ensure that the encoding of $\Htarget$ in the universal Hamiltonian is local, we make use of an idea originally from~\cite{universality_of_adiabatic_computation} and used recently in \cite{zhou:18,nirkhe:18,Aharonov_2014}, which has bee called `idling to enhance coherence'.
Before carrying out the phase-estimation computation, the system ``idles'' in its initial state for time $L$.
By choosing $L$ appropriately large, we can ensure that with high probability the input spins (the spins which form the unconstrained input $\ket{\psi}$ to $\MPE$) are found in their initial states. 
This means that (with high probability) there are a subset of spins on the simulator system whose state directly maps to the state which is being simulated in the target system.
This ensures that the encoding is (approximately) local (see \cref{lem:dovetailed-qpe} for detailed analysis of how idling is used to achieve universality).

\section{Universality} \label{universality}

\subsection{Translationally-Invariant Universal Models in 1D}
In this section we prove our main result: there exist translationally invariant, nearest neighbour Hamiltonians acting on a chain of qudits, which are universal quantum simulators.

All the `circuit-to-Hamiltonian' mappings we make use of in this work are what are known as ``Standard form Hamiltonians''.
Where ``Standard form Hamiltonians'' are a certain class of circuit-to-Hamiltonian constructions, defined in \cite{watson:19}.
We refer interested readers to \cite{watson:19} for the full definition - and simply note that it encompasses the Turing-machine based mappings which we make use of in this work \cite{spec-gap,Gottesman2009}.
In \cite{watson:19}, the following result was shown, which we will make use of in our proofs:
\begin{lemma}[{Standard form ground states; restatement of ~\cite[Lem.~5.8, Lem.~5.10]{watson:19}}] \label{james-lemma}
Let $H_\mathrm{SF}$ be a Standard Form Hamiltonian encoding a computation $U$, which takes (classical) inputs from a Hilbert space $\mathcal{S}$, and which sets an output flag with certainty if it is given an invalid input.
For $\ket{\psi_\mu} \in \mathcal{S}$ and $\Pi_{t=1}^T U_t = U$ we define
\[
\ket{\Phi(U,\psi_\mu)} \defas \frac{1}{\sqrt{T}} \sum_{t=1}^{T} U_t\ldots U_1 \ket{\psi_\mu}\ket{t}.
\]
Then $\mathcal{L} = \text{\emph{span}}\{ \ket{\Phi(U,\psi_\mu)} \}_{\mu=1}^{d^n}$ defines the kernel of $H_{SF}$, i.e.\ $H_\mathrm{SF}|_{\mathcal{L}} = 0$.
The\ smallest non-zero eigenvalue of $H_\mathrm{SF}$ scales as $1 - \cos{\pi / 2T}$.
\end{lemma}

We also require a digital quantum simulation algorithm, summarized in the following lemm:
\begin{lemma}[Implementing a Local Hamiltonian Unitary]\label{lem:sparse-sim}
	For a $k$-local Hamiltonian $H=\sum_{i=1}^m h_i$ on an $n$-partite Hilbert space of local dimension $d$, and where $m=\poly n$, there exists a QTM that implements a unitary $\tilde U$ such that
	\[
		\tilde U = \ee^{\ii H t} + \BigO(\epsilon),
	\]
	and which requires time $\poly(1/\epsilon, d^k, \| H \| t, n)$.
\end{lemma}
\begin{proof}
	Follows directly from \cite{Lloyd1996,Berry2005}.
\end{proof}
The polynomial time bound in \cref{lem:sparse-sim} suffices for our purposes; a tighter (and more complicated) bound, also for the more general case of sparse Hamiltonians, can be found in \cite{Berry2015}.

We can now start our main analysis by proving that ``dovetailing'' quantum computations---rigorously defined and constructed in~\cite[Lem.~22]{spec-gap}---can be used to construct universal simulators.
\begin{lemma}[Dovetailing for simulation]\label{lem:dovetailed-qpe}
Let $M_1$ be a QTM which writes out the binary expansion of some $x \in \field{N}$ on its work tape.
Assume there exists a standard form Hamiltonian which encodes the Turing machine $M_1$.
Then there also exists a standard form Hamiltonian $\HSF(x)$, which encodes the computation $M_1$ dovetailed with a QTM $\MPE$, such that the family of Hamiltonians
\begin{equation}
\Huniv(x) = \Delta \HSF(x) + T \sum_{i=0}^{N-1} \left(\sqrt{2} \Pi_{\alpha} - \Pi_{\beta} \right)
\end{equation}
can simulate any quantum Hamiltonian.
Here $\Delta$ and $T$ are parameters of the model, and $\Pi_{\alpha}$ and $\Pi_{\beta}$ are one-body projectors,
\end{lemma}

\begin{proof}[Proof of \Cref{lem:dovetailed-qpe}]
To prove this we show that the $\Huniv(x)$ can satisfy the definition to be an approximate simulation of an arbitrary ``target Hamiltonian" $\Htarget$, to any desired accuracy.
We break up the proof into multiple parts. 
First we construct a history state Hamiltonian $\HSF(x)$, which encodes two Turing machine computations: $M_1$ which extracts a description of $\Htarget$ from a parameter of $\HSF$, and $\MPE$ which carries out phase estimation on the unitary generated by $\Htarget$.
Then we define the one-body projectors $\Pi_\alpha$ and $\Pi_\beta$ which break up the ground space degeneracy of $\HSF$, and inflict just the right amount of penalty to approximately reconstruct the spectrum of $\Htarget$ in its entirety.

\paragraph{Construction of H\textsubscript{SF}.} \label{HSF_section}
$\HSF$ is a standard form history state Hamiltonian with a ground space laid out in \cref{james-lemma}.
The local states of the spins on which $\HSF$ acts are divided into multiple ``tracks''. There are a constant number of these, hence a constant local Hilbert space dimension. The exact number will depend on the standard form construction being used.
Each track serves its own purpose, as outlined in \cref{tab:local-hs}. See \cite{Gottesman2009,spec-gap} for more detail.

\begin{table}
\centering
\begin{tabular}{cl}
\toprule
Track & Purpose \\
\midrule
$1$ & Input track, contains input state $\ket{\psi} \in \mathbb{C}^2$ followed by string of $\ket{0}$s \\
\hline
$2$ & Turing machine work tape (shared by $M_1$ and $\MPE$ )\\
\hline
$3$ & Tape head and state for $M_1$ \\
\hline
$4$ & Tape head and state for $\MPE$ \\
\hline
$5,6,\ldots$ & Clock tracks for standard form clock construction \\
\bottomrule
\end{tabular}
\caption{Local Hilbert space decomposition for $\HSF$.}\label{tab:local-hs}
\end{table}

The QTM $\MPE$ reads in the description of $\Htarget$---provided as integer $x\in\field N$ output by the Turing machine $M_1$ whose worktape it shares.
$\MPE$ further reads in the unconstrained input state $\ket\psi$ (see \cref{tab:local-hs} for details of the local Hilbert space decomposition).
But instead of proceeding immediately, $\MPE$ idles for $L$ time-steps (where $L$ is specified in the input string $x$, as explained in \cref{sec:digital-encoding}), before proceeding to carry out the quantum phase estimation algorithm.

The quantum phase estimation algorithm is carried out with respect to the unitary $U = \ee^{\ii \Htarget\tau}$ for some $\tau$ such that $\| \Htarget\tau \| < 2\pi$.
It takes as input an eigenvector $\ket{u}$ of $U$, and calculates the eigenphase $\phi_u$.
The output of $\MPE$ is then the pair of integers $(a_u,b_u)$ (corresponding to the extracted phase $\phi_u=\sqrt2 a_u - b_u$ as explained in \cref{rem:2a-b}), specified in binary on an output track.
To calculate $\lambda_u$---the eigenvalue of $\Htarget$---to accuracy $\epsilon$ requires determining $\phi_u$ to accuracy $\BigO(\epsilon/\|\Htarget\|)$ which takes $\BigO(\|\Htarget\|/\epsilon)$ uses of $U=\ee^{\ii \Htarget \tau}$. The unitary $U$ must thus be implemented to accuracy $\BigO(\epsilon / \|\Htarget\|)$, which is done using \cref{lem:sparse-sim}; the latter introduces an overhead $\poly(n,d^k,\|\Htarget\|,\tau,1/\epsilon)$ in the system size $n$, local dimension $d$, locality $k$, and target accuracy $\epsilon$.
The error overhead of size $\poly1/\epsilon$ due to the digital simulation of the unitary is thus polynomial in the precision, as are the $\propto 1/\epsilon$ repetitions required for the QPE algorithm.
The whole procedure takes time
\begin{equation}\label{eq:QPE-TPE}
\TPE\defas \poly(d^k, \|\Htarget\|/\epsilon,n).
\end{equation}
 
In our construction the input to $\MPE$ is not restricted to be an eigenvector of $\ket{u}$, but it can always be decomposed as $\ket{\psi} = \sum_u m_u \ket{u}$.
By linearity, for input $\ket{\psi} = \sum_u m_u \ket{u}$ the output of $\MPE$ will be a superposition in which the output $(a_u,b_u)$ occurs with amplitude $m_u$.

After $\MPE$ has finished its computation, its head returns to the end of the chain.
A dovetailed counter then decrements $a_u, a_u-1, \ldots, 0$ and $b_u, b_u-1, \ldots, 0$.\footnote{For general input state $\ket{\psi} = \sum_u m_u \ket{u}$ there will be a superposition where the counter $a_u, a_u-1, \ldots, 0$ and $b_u, b_u-1, \ldots, 0$ occurs with amplitude $m_u$.}
For each timestep in the counter $a_u, a_u-1, \ldots, 0$ the Turing machine head changes one spin to a special flag state $\ket{\Omega_a}$ which does not appear anywhere else in the computation.
While for each timestep in the counter $b_u, b_u-1, \ldots, 0$ the Turing machine head changes one spin to a different flag state $\ket{\Omega_b}$.
(See e.g.~\cite[Lem.~16]{Bausch2018b}) for a construction of a Turing machine with these properties.)

By \cref{james-lemma}, the ground space $\mathcal L$ of $\HSF$ is spanned by computational history states as given in \cref{def:history-state}, and is degenerate since any input state $\ket\psi$ yields a valid computation.
Therefore:
\begin{equation}
\mathrm{ker}(\HSF) = \mathcal{L} =  \mathrm{span}_{\ket{\psi}}\left(\frac{1}{\sqrt{T}} \sum_{t=1}^{T}  \ket{\psi^{(t)}}\ket{t}\right)
\end{equation}
where $\ket{\psi^{(t)}}$ denotes the state of the system at time step $t$ if the input state was $\ket{\psi}$.

\paragraph{A Local Encoding.}
In order to prove that $\Huniv(N)$ can simulate all quantum Hamiltonians, we need to demonstrate that there exists a local encoding $\mathcal{E}(M)$ such that the conditions of \cref{app-sim} are satisfied.
To this end, let
\newcommand{\Phidle}{\Phi_\mathrm{idling}}
\[
\ket{\Phidle(\psi)} \defas \frac{1}{\sqrt{L'}} \sum_{t=1}^{L'} \ket{\psi^{(t)}}\ket{t}
\]
where $L' = T_1 + L$, and where $T_1$ is the number of time steps in the $M_1$ computation.
This is the history state up until the point that $\MPE$ begins its computation (i.e. the point at which the `idling to enhance coherence' ends).
So, throughout the computation encoded by this computation the spins which encode the information about the input state remain in their initial state, and we can write:
\[
\ket{\Phidle(\psi)} = \ket{\psi} \otimes \frac{1}{\sqrt{L'}} \sum_{t=1}^{L'} \ket{t}
\]
\newcommand{\Phicomp}{\Phi_\mathrm{comp}}
The rest of the history state we capture in
\[
\ket{\Phicomp(\psi)} \defas \frac{1}{\sqrt{T - L'}} \sum_{t=L'+1}^{T} \ket{\psi^{(t)}}\ket{t},
\]
such that the total history state is
\[
\ket{\Phi(\psi)} = \sqrt{\frac{L'}{T}}  \ket{\Phidle(\psi)} + \sqrt{\frac{T - L'}{T}}  \ket{\Phicomp(\psi)}.
\]

We now define the encoding $\mathcal{E}(M) = V M V^\dagger$ via the isometry
\begin{equation}
V = \sum_i \ket{\Phidle(i)}\bra{i}.
\end{equation}
where $\ket{i}$ are the computational basis states (any complete basis will suffice).
$\mathcal E$ is a local encoding, which can be verified by a direct calculation:
\newcommand{\phys}{^\mathrm{phys}}
\begin{equation}
\begin{split}
\mathcal{E}(A_j \otimes \identity) & = \sum_{ik}\ket{\Phidle(i)} \bra{i}(A_j \otimes \identity)\ket{k}\bra{\Phidle(k)} \\
& = \sum_{ik} \ket{i}\bra{i}(A_j \otimes \identity) \ket{k}\bra{k}   \otimes \frac{1}{L} \sum_{tt'=1}^L \ket{t}\bra{t'}  \\
& = (A_j \otimes \identity) \sum_{i} \ket{i}\bra{i}  \otimes \frac{1}{L} \sum_{tt'=1}^L \ket{t}\bra{t'} \\
& =\left( A\phys_j \otimes \identity\right) \sum_i\ket{\Phidle(i)}\bra{\Phidle(i)} \\
& = \left(A\phys_j \otimes \identity\right)\mathcal{E}(\identity),
\end{split}
\end{equation}
where $A\phys_j$ is the operator $A$ acting on the Hilbert space corresponding to the $j$\textsuperscript{th} qudit.

We now consider the encoding $\mathcal{E}'(M) = V'MV^{\prime\dagger}$, defined via
\begin{equation}
\label{eq:V'}
V' = \sum_i \ket{\Phi(i)}\bra{i}.
\end{equation}
We have that
\begin{equation}
\label{eq:V-V'}
\begin{split}
\|V' - V\|^2 & = \left\| \sum_i \left(\ket{\Phi(i)}\bra{i} - \ket{\Phidle(i)}\bra{i}\right) \right\| ^2\\
 & = \left\| \sum_i \left(\sqrt{\frac{T -L'}{T}}\ket{\Phicomp(i)}\bra{i} + \left(\sqrt{\frac{L'}{T}}-1\right)\ket{\Phidle(i)}\bra{i} \right) \right\|^2 \\
 &  \leq 2\left(1-\sqrt{\frac{L'}{T}}\right)\le 2\frac{T-L'}T=2\frac{\TPE}{T}.
\end{split}
\end{equation}
By \cref{james-lemma}, $S_{\mathcal{E}'}$ is the ground space of $\HSF$.

\paragraph{Splitting the Ground Space Degeneracy of H\textsubscript{\normalfont{\textbf{SF}}}.}
What is left to show is that there exist one body-projectors $\Pi_{\alpha}$ and $\Pi_{\beta}$ which add just the right amount of energy to states in the kernel $\mathcal L(\HSF)$ to reproduce the target Hamiltonian's spectrum.
We first choose the one body terms in $\Huniv$ to be projectors onto local subspaces which contain the two states which are outputs of the $\MPE$ computation - $\ket{\Omega_a}$ and $\ket{\Omega_b}$:
\[
    \Pi_a \defas \sum_{i=1}^N \ketbra{\Omega_a}_i
    \quad\text{and}\quad
    \Pi_b \defas \sum_{i=1}^N \ketbra{\Omega_b}_i.
\]

We have shown that if the input state is $\ket u$, which is an eigenstate of $U$ with eigenphase $\phi_u = a_u\sqrt{2} - b_u$, then the history state will contain $a_u$ terms with one spin in the state $\ket{\Omega_a}$ and $b_u$ terms with one spin in the state $\ket{\Omega_b}$ (each term in the history state will have amplitude $\frac{1}{T}$).
If the input is a general state $\ket\psi = \sum_u m_u \ket u$ then for each $u$ the history state will contain $a_u$ terms with one spin in the state $\ket{\Omega_a}$ and $b_u$ terms with one spin in the state $\ket{\Omega_b}$, where now each of these terms has amplitude $m_u / T$.

Let $\Pi \defas \sum_i \ket{\Phi(i)}\bra{\Phi(i)}$ for some complete basis $\ket{i}$,
and we define $H_1\defas T(\sqrt2 \Pi_a - \Pi_b)$, where $T$ is the total time in the computation.
It thus follows that the energy of $\ket{\Phi(u)}$ with respect to the operator $\Pi H_1 \Pi$ is given by $\phi_u  + \BigO(\epsilon)$.

Finally, we need the following technical lemma from \cite{BH17}.
\begin{lemma}[First-order simulation~\cite{BH17} ]
	\label{lem:firstorder}
	Let $H_0$ and $H_1$ be Hamiltonians acting on the same space and $\Pi$ be the projector onto the ground space of $H_0$. Suppose that $H_0$ has eigenvalue 0 on $\Pi$ and the next smallest eigenvalue is at least 1.
	Let $V$ be an isometry such that $VV^{\dagger}=\Pi$ and
	\begin{equation}
	\label{eq:firstorderrequirement}
	\|V \Htarget V^\dag - \Pi H_1 \Pi\| \le \epsilon/2.
	\end{equation}
	Let $H_{\operatorname{sim}} = \Delta H_0 + H_1$ .
	Then there exists an isometry $\tilde{V}$ onto the the space spanned by the eigenvectors of $H_{\operatorname{sim}}$ with eigenvalue less than $\Delta/2$ such that
	\begin{enumerate}
		\item $\|V-\tilde{V}\| \le \BigO(\|H_1\|/\Delta)$
		\item $\|\tilde{V}H_{\operatorname{target}} \tilde{V}^{\dagger} -H_{\operatorname{sim}< \Delta/2} \| \le \epsilon/2 + \BigO(\|H_1\|^2/\Delta)$
	\end{enumerate}
\end{lemma}
We will apply Lemma~\ref{lem:firstorder} with $H_0=2T^2\HSF$ and $H_1=T(\sqrt2 \Pi_a - \Pi_b)$.
We have $\lambda_{\min}( \HSF) = 0$ and the next smallest non-zero eigenvalue of $\HSF$ is $(1-\cos(\pi/2T)\ge 1/2T^2)$ by \cref{james-lemma}, so $H_0=2T^2\HSF$ has next smallest non-zero eigenvalue at least 1.
Moreover, $\left\|H_1\right\| = \sqrt{2}T$.
Note that $V'$, as defined in \cref{eq:V'}, is an isometry which maps onto the ground state of $H_0$.
By construction we have that the spectrum of $\Htarget$ is approximated to within $\epsilon$ by $H_1$ restricted to the ground space of $\HSF$, thus $\|\Pi H_1 \Pi - \tilde{\mathcal{E}}(H)\| \leq \epsilon$.

Lemma~\ref{lem:firstorder} therefore implies that there exists an isometry $\tilde{V}$ that maps exactly onto the low energy space of $\Huniv$ such that $\|\tilde{V}-V'\|\le \BigO(\sqrt{2}T/(\Delta/2T^2))=\BigO(T^3/\Delta)$.
By the triangle inequality and \cref{eq:V-V'}, we have:
\begin{equation}\label{eq:T-scaling}
\|V-\tilde{V}\|\le \|V-V'\|+\|V'-\tilde{V}\|\le O \left(\frac{T^3}{\Delta} + \frac{\TPE}{T}\right).
\end{equation}

The second part of the lemma implies that
\begin{equation}
\|\tilde{V} \Htarget \tilde{V}^{\dagger} -H_{\operatorname{univ} <\Delta'/2}\| \le \epsilon/2+ \BigO((\sqrt{2}T)^2/(\Delta/2T^2))=\epsilon/2 +\BigO(T^4/\Delta).
\end{equation}
Therefore, the conditions of \cref{app-sim} are satisfied for a $(\Delta',\eta,\epsilon')$-simulation of $\Htarget$, with $\eta = O \left(T^3/\Delta + \TPE /T\right)$, $\epsilon' = \epsilon+\BigO(T^4 / \Delta)$ and $\Delta'= \Delta/2T^2$.
Therefore we must increase $L$ so that $T \ge \BigO(\TPE /\eta)=\poly(n, d^k, \|H\|,1/\epsilon,1/\eta)$ by \cref{eq:QPE-TPE}, (thereby determining $x$), and increase $\Delta$ so that 
\begin{equation}
\label{eq:Deltascaling}
\Delta \ge \Delta' T^2 +\frac{T^3}{\eta}+\frac{T^4}{\epsilon}
\end{equation}
to obtain a $(\Delta', \eta, \epsilon)$-simulation of the target Hamiltonian.
The claim follows.
\end{proof}

We can now prove our main theorem:

\begin{theorem} \label{main-theorem-1}
There exists a  two-body interaction depending on a single parameter $h(\phi)$ such that the family of translationally-invariant Hamiltonians on a chain of length $N$,
\begin{equation}
\Huniv(\phi, \Delta, T) = \Delta \sum_{\langle i,j \rangle} h(\phi)_{i,j} + T \sum_{i=0}^{N-1} \left(\sqrt{2} \Pi_{\alpha} - \Pi_{\beta} \right)_i,
\end{equation}
is a universal model, where $\Delta$, $T$ and $\phi$ are parameters of the Hamiltonian, and the first sum is over adjacent site along the chain.
Furthermore, the universal model is efficient in terms of the number of spins in the simulator system.
\end{theorem}
\begin{proof}

The two body interaction $h(\phi)$ makes up a standard form Hamiltonian which encodes a QTM, $M_1$ dovetailed with the phase-estimation computation from \cref{lem:dovetailed-qpe}.
The QTM $M_1$ carries out phase estimation on the parameter $\phi$ in the Hamiltonian, and writes out the binary expansion of $\phi$ (which contains a description of the Hamiltonian to be simulated) on its work tape.
There is a standard form Hamiltonian in~\cite{spec-gap} which encodes this QTM, so by \cref{lem:dovetailed-qpe} we can construct a standard form Hamiltonian which simulates all quantum Hamiltonians by dovetailing $M_1$ with $\MPE$.

The space requirement for the computation is $\BigO(|\phi|)$, where $|\phi|$ denotes the length of the binary expansion of $\phi$, and the computation requires time $T_1 = \BigO(|\phi|2^{|\phi|})$ \cite[Theorem 10]{spec-gap-full}
As we commented in \cref{sec:spec-gap}, the standard form clock construction set out in \cite[Section 4.5]{spec-gap-full} allows for computation time of $ \BigO(|\phi|2^{|\phi|})$ using a Hamiltonian on $|\phi|$ spins.
We therefore find that for a $k$-local target Hamiltonian $\Htarget$ acting on $n$ spins of local dimension $d$, the number of spins required in the simulator system for a simulation that is $\epsilon$ close to $\Htarget$ is given by $N =  \BigO(|\phi|) =\poly\left(n,d^k,\|H\|,1/\eta,1/\epsilon \right)$.

Therefore, the universal model is efficient in terms of the number of spins in the simulator system as defined in \cref{def:efficient-sim}.
\end{proof}

Note that this universal model is \emph{not} efficient in terms of the norm $\|\Huniv\|$.
This is immediately obvious, since $\| \Huniv \| = \Omega(\Delta)$, and using the relations between  $\Delta'$, $\eta$, $\epsilon$, and  $T$ and $\Delta$ from \cref{lem:dovetailed-qpe,eq:Deltascaling},
\[T=T_1+L+\TPE=O\left(2^x+\poly\left(n,d^k,\| \Htarget \|, \frac1\epsilon, \frac1\eta \right)\right)
\quad \text{ and } \quad \Delta \ge \Delta' T^2 +\frac{T^3}{\eta}+\frac{T^4}{\epsilon}\]
by \cref{eq:QPE-TPE}, so $T,\Delta$ are both $\poly\left(2^x,\|\Htarget\|,\Delta',1/\epsilon, 1/\eta \right)$.
For a $k$-local Hamiltonian $\Htarget$ with description $x$ as presented in \cref{sec:digital-encoding}, $|x|=\Omega\left(md^{2k} \log(\|\Htarget\|md^{2k}/\delta)\right)$.

However if we only wish to simulate a translationally invariant $k$-local Hamiltonian $\Htarget$, this can be specified to accuracy $\delta$ with just $\log(\|\Htarget\|m d^{2k} /\delta)$ bits of information.
In this case (for $d,k=\BigO(1)$ and taking $\delta=\epsilon$), the interaction strengths are then $\poly(n,\|\Htarget\|,\Delta', \frac{1}{\eta},\frac{1}{\epsilon})$, and the whole simulation is efficient.

\Cref{lem:dovetailed-qpe} also allows the construction of a universal quantum simulator with two free parameters.
\begin{theorem}\label{main-theorem-2}
There exists a fixed two-body interaction $h$ such that the family of translationally-invariant Hamiltonians on a chain of length $N$,
\begin{equation}
\Huniv(\Delta, T) = \Delta \sum_{\langle i,j \rangle} h_{i,j} + T \sum_{i=0}^{N-1} \left(\sqrt{2} \Pi_{\alpha} - \Pi_{\beta} \right)_i,
\end{equation}
is a universal model, where $\Delta$ and $T$ are parameters of the Hamiltonian, and the first sum is over adjacent sites along the chain.
\end{theorem}
\begin{proof}
As in \cref{main-theorem-1}, the two body interaction $h$ makes up a standard form Hamiltonian which encodes a QTM $M_1$ dovetailed with the phase-estimation computation from \cref{lem:dovetailed-qpe}.
It is based on the construction from~\cite{Gottesman2009}.

Take $M_1$ to be a binary counter Turing machine which writes out $N$---the length of the qudit chain---on its work tape. We will choose $N$ to contain a description of the Hamiltonian to be simulated, as per \cref{sec:digital-encoding}.
There is a standard form Hamiltonian in~\cite{Gottesman2009} which encodes this QTM, so by \cref{lem:dovetailed-qpe} we can construct a standard form Hamiltonian which simulates all quantum Hamiltonians by dovetailing $M_1$ with $\MPE$.

Since $B(N)$, as defined in \cref{eq:B(x)}, contains a description of the Hamiltonian to be simulated, we have that
\[
N = \poly\left(2^{\poly(n,\|\Htarget\|,1/\eta,1/\epsilon)} \right).
\]
The standard form clock used in the construction allows for computation time polynomial in the length of the chain, so $\exp(\poly)$-time in the size of the target system.
As before, by \cref{eq:QPE-TPE},  we require \[T=T_1+L+\TPE=O\left(N+\poly\left(n,d^k,\| \Htarget \|, \frac1\epsilon, \frac1\eta \right)\right)
\quad \text{ and } \quad \Delta \ge \Delta' T^2 +\frac{T^3}{\eta}+\frac{T^4}{\epsilon}.\]
\end{proof}

According to the requirements of \cref{app-sim}, the universal simulator of the second theorem is not efficient in either the number of spins, nor in the norm.
However---as was noted in~\cite{PiddockBausch}---this is unavoidable if there is no free parameter in the universal Hamiltonian which encodes the description of the target Hamiltonian:
a translationally invariant Hamiltonian on $N$ spins can be described using only $\BigO(\poly\log(N))$  bits of information, whereas a $k$-local Hamiltonian which breaks translational invariance in general requires $\poly(N)$ bits of information.
So, by a simple counting argument, we can see that it is not possible to encode all the information about a $k$-local Hamiltonian on $n$ spins in a fixed translationally invariant Hamiltonian acting on $\poly(n)$ spins.

We observe that the parameters $\Delta$ and $T$ are qualitatively different to $\phi$, in that they do not depend on the Hamiltonian to be simulated, but only the parameters $(\Delta',\epsilon,\eta)$ determining the precision of the simulation.
%
%
%

\subsection{No-Go for Parameterless Universality}
Is an explicit $\Delta$-dependence of a simulator Hamiltonian $\Huniv$ necessary to construct a universal model?
Note that an implicit dependence of $\Huniv$ on $\Delta$ is possible via the chain length $N=N(\Delta)$ in \cref{main-theorem-1}.
In the following, we prove that such an implicit dependence is insufficient, by giving a concrete counterexample for which an explicit $\Delta$-dependence is necessary.

To this end, we note that it has previously been shown~\cite{Aharonov2018a} that a degree-reducing Hamiltonian simulation (in a weaker sense of simulation, namely gap-simulation where only the ground state(s) and spectral gap are to be maintained) is only possible if the norm of the \emph{local} terms is allowed to grow.
In order to construct a concrete example in which an explicit $\Delta$-dependence is necessary, we first quote~\citeauthor{Aharonov2018a}'s result, and then translate the terminology to our setting.
\begin{theorem}[{\citeauthor{Aharonov2018a} (\cite[Thm.~1]{Aharonov2018a})}]\label{th:counterex}
For sufficiently small constants $\epsilon\ge0$ and $\tilde{\omega}\ge0$, there exists a minimum system size $N_0$ such that for all $N\ge N_0$ there exists no constant-local $[r,M,J]=[\BigO(1),M,\BigO(1)]$ gap simulation (where $r$ is the interaction degree, $M$ the number of local terms, and $J$ the local interaction strength of the simulator) of the Hamiltonian
\[
    H_A \defas \frac14 \sum_{i=1}^N \sum_{j<i} (1-\sigma_z^{(i)}) \otimes (1-\sigma_z^{(j)}) = \sum_{i=1}^N \sum_{j<i}\ketbra{1}^{(i)} \otimes \ketbra{1}^{(j)}
\]
with a localized encoding, $\epsilon$-incoherence, and energy spread $\tilde{\omega}$, for any number of Hamiltonian terms $M$.
\end{theorem}

\begin{corollary}\label{cor:counterex}
Consider a universal family of Hamiltonians with local interactions and bounded-degree interaction graph.
Hamiltonians in this family must have an explicit dependence on the energy cut-off ($\Delta$) below which they are valid simulations of particular target Hamiltonians.
\end{corollary}
\begin{proof}
We first explain the notation used in \cref{th:counterex}.
As mentioned, the notion of gap simulation is weaker than \cref{app-sim}. Only the (quasi-) ground space $\mathcal L$ of $H_A$, rather than the full Hilbert space, needs to be represented $\epsilon$-coherently: $\| H_A|_{\mathcal L} - \tilde H_A|_{\mathcal L}\| < \epsilon$, where $\cdot|_{\mathcal L}$ denotes the restriction to $\mathcal L$). And only the spectral gap above the ground space, rather than the full spectrum, must be maintained: $\tilde{\gamma}=\Delta(\tilde H_A) \ge \gamma = \Delta(H_A)$.
The rest of the spectrum in the simulation can be arbitrary.
Energy spread in this context simply means the range of eigenvalues within $\mathcal L$ spreads out at most such that $|\lambda_0 - \tilde{\lambda}_0|\le \tilde \omega\gamma$.

A $[\BigO(1),M,\BigO(1)]$ simulation with the above parameters then simply means an $\epsilon$-coherent gap simulation, constant degree and local interaction strength, where $M$---the number of local terms in the simulator---is left unconstrained, and the eigenvalues vary by at most $\tilde{\omega}\gamma$.

It is clear that this notion of simulation falls within our more generic framework of simulation (cf.~\cite[Sec.~1.1]{Aharonov2018a}): a simulation of $H_A$ \emph{also} defines a valid gap simulation of $H_A$. Since by \cref{def:efficient-sim} this simulation can be made arbitrarily precise, with parameters $\epsilon,\tilde\omega$ arbitrarily small, and has constant interaction degree by assumption, this contradicts \cref{th:counterex}.
\end{proof}

\section{Applications to Hamiltonian Complexity}\label{sec:complexity}

As already informally stated, the \lham problem is the question of approximating the ground state energy of a local Hamiltonian to a certain precision.
Based on a history state embedding of a \QMA verifier circuit and on Feynman's circuit-to-Hamiltonian construction~\cite{Feynman1986}, Kitaev proved in~\citeyear{Kitaev2002} that \lham with a promise gap that closes inverse-polynomially in the system size is \QMA-complete~\cite{Kitaev2002}.

To be precise, let us start by defining the \lham problem. We note that variants of this definition can be found throughout literature which commonly omit one or more of the constraints presented herein, in particular with regards to the bit precision to the input matrices.
In order to be precise, we explicitly list the matrix entries' bit precision as extra parameter $\Sigma$ in the following definition.
\begin{problem}[\lham$(f,\Sigma)$]
\probleminput{
Local Hamiltonian $H=\sum_{i=1}^m h_i$ on an $N$-partite Hilbert space of constant local dimension, and $m\le \poly (N)$.
Each $h_i \defas h_{S_i} \otimes \1_{S_i^c}$ acts non-trivially on at most $|S_i| \le k$ sites, and $\| h_i \| \le 1$.
Two numbers $\alpha,\beta>0$.
The bit complexity of the matrix entries of $h_i$ is $\BigO(\Sigma(N))$.
}
\problempromise{
$\beta-\alpha \ge f(N)$, and $\lmin(H)$ either $\ge\beta$, or $\le\alpha$.
}
\problemquestion{
\YES\ if  $\lmin(H) \ge \beta$, else \NO.
}
\end{problem}
Kitaev's \QMA-completeness result was shown for a promise gap $f(N) = \poly N$~\cite[Th.~14.1]{Kitaev2002}. Following the proof construction therein reveals that this was done for a bit complexity of the matrix entries $\Sigma(N)=\BigO(1)$ (assuming a discrete fixed gateset for the encoded \QMA verifier).
Since his seminal result, the statement has been extended and generalized to ever-simpler many-body systems~\cite{Oliveira2008,Hallgren2013,Aharonov2009}. Some of these results allow a coupling constant to scale in the system size, e.g.\ as $\poly N$---i.e.\ the matrix entries now feature a bit precision of $\Sigma(N)=\poly\log N$.

We remark that despite the apparent relaxation in the bit precision, these results are \emph{not} weaker than~\citeauthor{Kitaev2002}'s.
Since the number of local terms $m=\poly N$, a polynomial number of local terms of $\BigO(1)$ bit complexity acting on the same sites can already be combined to create $k$-local interactions with polynomial precision (logarithmic bit-precision, $\Omega(1/\poly)\cap\BigO(\poly)$). (Similar to how the encoding in \cref{sec:digital-encoding,rem:2a-b} works by adding up integers to approximate a number in the interval $[0,1]$.) We also emphasize that the overall bit complexity of the input is already $\poly N$, as there are that many local terms to specify in the first place. Indeed, many times in the literature, the matrix entries of the \lham problem are simply restricted to bit precision $\Sigma=\poly N$ (e.g.~\cite{Cubitt2013}).

However, translationally-invariant spin systems are common in condensed matter models of real-world materials, whereas models with precisely-tuned interations that differ from site to site are less realistic.
It is known that \QMA-hardness of approximating the ground state energy to $1/\poly$ precision in the system size is a property of non-translationally-invariant couplings, that prevails even when those couplings are arbitrarily close to identical~\cite[Cor.~21]{Bausch2018c}. But even small amounts of disorder can radically change the properties of quantum many-body systems compared to strict translational invariance, which is the intuition behind this result. A variant of \lham for the strictly translationally-invariant case can be formulated as follows:

\begin{problem}[\tilham$(f,\Sigma)$]
\probleminput{
Translationally-invariant\footnote{Naturally, translational invariance is defined with respect to the Hilbert space's interaction graph on $\Lambda$.}
 local Hamiltonian $H=\sum_{i\in\Lambda} h_i$ on an $N$-partite Hilbert space $(\field C^d)^{\otimes\Lambda}$ of constant local dimension $d$.
Each $h_i \defas (h)_{S_i} \otimes \1_{S_i^c}$ for some fixed hermitian operator $h$ acts non-trivially and in a translationally-invariant fashion on at most $|S_i|\le k$ sites, and $\| h_i \| \le 1$.
Two numbers $\alpha,\beta>0$.
The bit complexity of the matrix entries of $h_i$ is $\BigO(\Sigma(N))$.
}
\problempromise{
$\beta-\alpha \ge 1/f(N)$, and $\lmin(H)$ either $\ge\beta$, or $\le\alpha$.
}
\problemquestion{
\YES if  $\lmin(H) \ge \beta$, else \NO.
}
\end{problem}

\citeauthor{Gottesman2009} proved in~\citeyear{Gottesman2009} that \tilham$(\poly,1)$ is \QMAEXP-complete~\cite{Gottesman2009}, which has since been generalized to systems with lower local dimension~\cite{Bausch2016,Bausch2017}, variants of which again introduce a polynomially-scaling local coupling strength.
We emphasize that while~\citeauthor{Gottesman2009}'s definition restricts the bit precision $\Sigma$ to be constant, the input size to the problem---namely the chain length $N$---is already of size $\log N$. A poly-time reduction thus does not change the complexity class, and allowing matrix entries of size $\poly\log N$ is arguably natural. As noted in~\cite[Sec.~3.3]{Bausch2016}, an equivalent definition for \tilham can thus be obtained by relaxing the norm of the local terms to $\| h_i \| \le \poly N$, given the promise gap $f(N) = \Omega(\poly N)$.

Care has to be taken in defining \QMAEXP for the right input scaling. For \tilham$(\poly,1)$, the input size is given by the system size only, as all the local terms are specified by a constant number of bits.
This means that \tilham$(\poly,1)$ is indeed \QMAEXP hard, \emph{but for an input of size $\lceil\log(N)\rceil$, where $N$ is the size of the system}.
As Karp reductions are allowed for \QMAEXP, this does not change if we allow the local terms to scale polynomially in the system size; the problem input is still of size at most $\poly\log$, and thus constitutes a well-defined input for \QMAEXP with respect to this input size.
Informally, \QMAEXP(``$\poly \log(N)$-sized input'')~$<$~\QMA(``$\poly N$-sized input''), as only that scaling allows to both saturate and maintain the $1/\poly$ promise gap.
In short, the problem is \emph{easier} for translationally-invariant systems, as expected. (We refer the reader to the extended discussion in~\cite[Sec.~3.4]{Bausch2016}.)


How does the situation change if we allow a promise gap that scales differently? In particular, how hard is \lham$(\exp\poly)$?
In~\cite{Fefferman2016} the authors characterize this setup, which they use for a reduction from \PreciseQMA. The \PreciseQMA verifier has a $1/\exp\poly$ promise gap, instead of \QMA's usual $1/\poly$ promise gap. (Note that it is this very promise gap which naturally maps to the \lham's promise gap on the ground state energy.) They show that \lham$(1/\exp\poly)$ is complete for \PreciseQMA, which they further show equals \PSPACE.
We emphasize that the authors did not explicitly restrict the bit precision. Yet a natural restriction in this context is again $\Sigma(N) = \poly N$, as there are $m=\poly N$ local terms to specify. And a larger bit precision makes the input size too large for containment in \PreciseQMA.

A natural question to ask is thus: how hard is \tilham$(\exp\poly,\Sigma(N))$ for either $\Sigma(N)=\poly N$ or $\poly\log N$? Furthermore, is it easier because of the translational invariance, as it was for the $\poly$-promise-gap case?
We show that this is \emph{not} the case, and prove the following result.
\begin{theorem}\label{th:pspace}
\tilham$(\exp\poly,\poly)$ is \PSPACE-complete.
\end{theorem}
\begin{proof}
The result follows by \cref{main-theorem-1}. Specifying all the local terms in $H$ requies an exponentially long QPE computation to extract $\poly(N)$ many bits from a phase. Because a \PreciseQMA-complete local Hamiltonian $H$ already has a $1/\exp\poly(N)$-closing promise gap, this does not attenuate the resulting promise gap by more than another exponential factor.
Containment in \PSPACE follows by~\cite{Fefferman2016}.
\end{proof}

\Cref{th:pspace} illustrates a curious mismatch: irrespective of the promise gap scaling or matrix bit precision, \tilham features the system size $N$ as input. A $1/\poly N$ promise gap and $\poly \log N$ bit precision saturate this input, and yield a \QMAEXP-complete construction, as discussed above.
Yet when we need to specify a $1/\exp N$ promise gap, \emph{that} bit precision is the dominant input. So we might as well specify the local terms to the same $\poly N$ bit precision, which in turn allows the translationally-invariant system to simulate a non-translationally-invariant one.


\section{Applications to Holography} \label{holography-implications}

We can use the universal Hamiltonian constructions in this paper to construct a 2D-to-1D holographic quantum error correcting code (HQECC) with a local boundary Hamiltonian.
HQECCs are toy models of the AdS/CFT correspondence which capture many of the qualitative features of the duality~\cite{Pastawski:2015,osborne-17,Hayden:2016}.
Recently, a HQECC was constructed from a 3D bulk to a 2D boundary which mapped local Hamiltonians in the bulk to local Hamiltonians in the boundary~\cite{HQECC-local}.
The techniques in~\cite{HQECC-local} require at least a 2D boundary, and it was an open question whether a similar result could be obtained in lower dimensions.

Here we construct a HQECC from a 2D bulk to a 1D boundary which maps any (quasi-)local Hamiltonian in the bulk to a local Hamiltonian in the boundary.
A quasi $k$-local Hamiltonian is a generalisation of a $k$-local Hamiltonian, where instead of requiring that each term in the Hamiltonian acts on only $k$-spins, we require that each term in the Hamiltonian has Pauli rank at most $k$,\footnote{The Pauli rank of an operator is the number of terms in its Pauli decomposition.} along with some geometric restrictions on the interaction graph.
More precisely:

\begin{definition}[Quasi-local hyperbolic Hamiltonians]
  Let $\hyper$ denote $d$-dimensional hyperbolic space, and let $B_r(x)\subset\hyper$ denote a ball of radius $r$ centred at $x$.
  Consider an arrangement of $n$ qudits in $\hyper$ such that, for some fixed $r$, at most $k$ qudits and at least one qudit are contained within any $B_r(x)$.
  Let $Q$ denote the minimum radius ball $B_Q(0)$ containing all the qudits (which without loss of generality we can take to be centred at the origin).
  A quasi $k$-local Hamiltonian acting on these qudits can be written as:
  \begin{equation}
  \Hbulk = \sum_Z h^{(Z)}
  \end{equation}
  where the sum is over the $n$ qudits, and each term can be written as:
  \begin{equation}
  h^{(Z)} = h_{\mathrm{local}}^{(Z)} h_{\mathrm{Wilson}}^{(Z)}
  \end{equation}
  where:
  \begin{itemize}
  \item $h_{\mathrm{local}}^{(Z)}$ is a term acting non-trivially on at most $k$ qudits which are contained within some $B_r(x)$
  \item $h_{\mathrm{Wilson}}^{(Z)}$ is a Pauli operator acting non trivially on at most $\BigO(L-x)$ qudits which form a line between $x$ and the boundary of $B_Q(0)$
  \end{itemize}
\end{definition}

The extension to quasi-local bulk Hamiltonians allows us to consider using the HQECC to construct toy models of AdS/CFT with gravitational Wilson lines in the bulk theory.\footnote{Although in~\cite{HQECC-local} the result is only proved for local Hamiltonians, the proof can trivially be extended to encompass quasi-local bulk Hamiltonians in the 3D-2D setting too.}

\noindent With this definition, we obtain the following result.
\begin{theorem} \label{holography}
Consider any arrangement of $n$ qudits in $\hyper$, such that for some fixed $r$ at most $k$ qudits and at least one qudit are contained within any $B_r(x)$.
  Let $Q$ denote the minimum radius ball $B_Q(0)$ containing all the qudits.
  Let $\Hbulk = \sum_Z h_Z$ be any (quasi) $k$-local Hamiltonian on these qudits.

  Then we can construct a Hamiltonian $\Hboundary$ on a 1D boundary manifold $\mathcal{M}$ with the following properties:
  \begin{enumerate}
  \item%
    $\mathcal{M}$ surrounds all the qudits and has diameter $\BigO\left(\max\left(1,\log(k)/r\right) Q + \log\log n\right)$.
      \item%
    The Hilbert space of the boundary consists of a chain of qudits of length $\BigO\left(n\log n \right)$.
  \item%
    Any local observable/measurement $M$ in the bulk has a set of corresponding observables/measurements $\{M'\}$ on the boundary with the same outcome. A local bulk operator $M$ can be reconstructed on a boundary region $A$ if $M$ acts within the greedy entanglement wedge of $A$, denoted $\mathcal{E}[A]$.\footnote{The entanglement wedge, $\mathcal{E}_A$ is a bulk region constructed from the minimal area surface used in the Ryu-Takayanagi formula. It has been suggested that on a given boundary region, $A$, it should be possible to reconstruct all operators which lie in $\mathcal{E}_A$~\cite{Headrick:2014}. The greedy entanglement wedge is a discretised version defined in~\cite[Definition~8]{Pastawski:2015}}
  \item%
    $\Hboundary$ consists of 2-local, nearest-neighbour interactions between the boundary qudits.
  \item%
    $\Hboundary$ is a $(\Delta_L,\epsilon,\eta)$-simulation of $\Hbulk$ in the sense of \cref{app-sim}, with $\epsilon,\eta = 1/\poly(\Delta_L)$, $\Delta_L = \Omega\left(\|\Hbulk\|\right)$, and where the interaction strengths in $\Hboundary$ scale as  $\max_{ij}|\alpha_{ij}| = \BigO\left( (\Delta_L + 1/\eta + 1/\epsilon) \poly(e^R 2^{e^R})\right)$.
  \end{enumerate}
\end{theorem}

\begin{proof}
  There are three steps to this simulation. The first two steps follow exactly the same procedure as in~\cite{HQECC-local}.

\paragraph{Step 1.} Simulate $\Hbulk$ with a Hamiltonian which acts on the bulk indices of a HQECC in $\hyper$ of radius $R = \BigO\left(\max\left(1,\log(k)/r\right) L\right)$.

  In order to do this, we embed a tensor network composed of perfect tensors in a tessellation of $\hyper$ by a Coxeter polygon with associated Coxeter system $(W,S)$, and growth rate $\tau$.
  Note that in a tessellation of $\hyper$ by Coxeter polytopes the number of polyhedral cells in a ball of radius $r'$ scales as $\BigO(\tau^{r'})$, where we are measuring distances using the word metric, $d(u,v) = l_S(u^{-1}v)$. (See~\cite{HQECC-local} for a detailed discussion.)

  If we want to embed a Hamiltonian $\Hbulk$ in a tessellation we will need to rescale distances between the qudits in $\Hbulk$ so that there is at most one qudit per polyhedral cell of the tessellation.
  If $\tau^{r'} = k$, then
  \[
    \frac{r'}{r} = \frac{\log(k)}{\log(\tau) r} = \BigO\left(\frac{\log(k)}{r}\right).
  \]
  If $\log(k)/r\geq 1$ then the qudits in $\Hbulk$ are more tightly packed than the polyhedral cells in the tessellation, and we need to rescale the distances between the qudits by a factor of $\BigO\left(\log(k)/r\right)$.
  If $\log(k)/r < 1$ then the qudits in $\Hbulk$ are less tightly packed then the cells of the tessellation, and there is no need for rescaling.
  The radius $R$ of the tessellation needed to contain all the qudits in $\Hbulk$ is then given by
  \begin{equation}
    R =
    \begin{cases}
      \BigO\left(\log(k)/rL\right),& \text{if } \log(k)/r \geq 1 \\
      \BigO(L) & \text{otherwise}.
    \end{cases}
  \end{equation}

  After rescaling there is at most one qudit per cell of the tessellation.
  There will be some cells of the tessellation which do not contain any qudits.
  We can put ``dummy'' qudits in those cells which do not participate in any interactions, so their inclusion is just equivalent to tensoring the Hamiltonian with an identity operator.
  We can upper and lower bound the number of ``real'' qudits in the tessellation.
  If no cells contain dummy qudits then the number of real qudits in the tesselation is given by $n_{\max} = N = \BigO(\tau^R)$, where $N$ is the number of cells in the tessellation.
  By assumption, there is at least one real qudit in a ball of radius $r'$. Thus the minimum number of real qudits in the tessellation scales as $n_{\min} = \BigO(\tau^R / \tau^{r'} ) = \BigO(\tau^R) = \BigO(N)$, and $n = \Theta(\tau^R) = \Theta(N)$.

  If the tessellation of $\hyper$ by Coxeter polytopes is going to form a HQECC, the Coxeter polytope must have at least 5 faces~\cite[Theorem 6.1]{HQECC-local}.
  From the HQECC constructed in~\cite{Pastawski:2015} it is clear that this bound is achievable, so we will without loss of generality assume the tessellation we are using is by a Coxeter polytope with 5 faces.
  The perfect tensor used in the HQECC must therefore have 6 indices.

  It is known that there exist perfect tensors with 6 indices for all local dimensions $d$~\cite{Rains:1997}.
  We will restrict ourselves to stabilizer perfect tensors with local dimension $p$ for some prime $p$. These can be constructed for $p=2$~\cite{Pastawski:2015} and $p \geq 7$~\cite{Helwig:2013a}.
  Qudits of general dimension $d$ can be incorporated by embedding qudits into a $d$-dimensional subspace of the smallest prime which satisfies $p \geq d$ and $p=2$ or $p \geq 7$.
  We then add one-body projectors onto the orthogonal complement of these subspaces, multiplied by some $\Delta_S' \geq |\Hbulk|$ to the embedded bulk Hamiltonian.
  The Hamiltonian $\Hbulk'$ on the $n$ $p$-dimensional qudits is then a perfect simulation of $\Hbulk$.

 We can therefore simulate any $\Hbulk$ which meets the requirements stated in the theorem with a Hamiltonian which acts on the bulk indices of a HQECC in $\hyper$.

\paragraph{Step 2.} Simulate $\Hbulk$ with a Hamiltonian $H_B$ on the boundary surface of the HQECC.

We first set $H_B \coloneqq  H' + \Delta_SH_S$, where $H'$ satisfies $H'\Pi_{\mathcal{C}} = V(\Hbulk' \otimes \identity_\mathrm{dummy})V^\dagger$.
Here $V$ is the encoding isometry of the HQECC, $\Pi_{\mathcal{C}}$ is the projector onto the code-subspace of the HQECC, $\identity_\mathrm{dummy}$ acts on the dummy qudits and $H_S$ is given by
 \begin{equation}
    H_S \coloneqq \sum_{w \in W} \left( \identity - \Pi_{\mathcal{C}^{(w)}}\right).
 \end{equation}
 $\Pi_{\mathcal{C}^{(w)}}$ is the projector onto the codespace of the quantum error correcting code defined by viewing the $w$\textsuperscript{th} tensor in the HQECC as an isometry from its input indices to its output indices (where input indices are the bulk logical index, plus legs connecting the tensor with those in previous layers of the tessellation).

  Provided $\Delta_S \geq \|\Hbulk'\|$,~\cite[Lemma 6.9]{HQECC-local} ensure that $H_B$ meets the conditions in \cref{def:exact-sim} to be a perfect simulation of $\Hbulk'$ below energy $\Delta_S$, and hence---as simulations compose---a perfect simulation of $\Hbulk$.

  Naturally, there is freedom in this definition as there are many $H'$ which satisfy the condition stated.
  We will choose an $H'$ where every bulk operator has been pushed out to the boundary, so that a $1$-local bulk operator at radius $x$ corresponds to a boundary operator of weight $\BigO(\tau^{R-x})$.
  We will also require that the Pauli rank of every bulk operator has been preserved (see~\cite[Theorem D.4]{HQECC-local} for proof we can choose $H'$ satisfying this condition).

\paragraph{Step 3.} Simulate $H_B$ with a local, nearest neighbour Hamiltonian using the technique from \cref{main-theorem-1}.

  In order to achieve the scaling quoted we make use of the structure of $H_B$ due to the HQECC.
 It can be shown~\cite{HQECC-local} that $H_B$ will contain $\BigO(\tau^x)$ Pauli rank-1 operators of weight $\tau^{R-x}$ for $0 \leq x \leq R$. A Pauli rank-1 operator of weight $w$ can be specified using $\BigO(w)$ bits of information. So, if we encode $H_B$ in the binary expansion of $\phi$ as
 \[
 B(\phi) = \gamma'(R) \cdot_{x=0}^R \left[\gamma'(i)^{\cdot \tau ^{R-x}} \cdot \left(\gamma'(a_j) \cdot \gamma'(b_j) \cdot P_1 \cdot \ldots  \cdot P_{\tau^{R-x}} \right)  \right]^{\cdot \tau^x} \cdot \gamma'(L),
 \]
we have $|\phi| = \BigO(R\tau^R) = \BigO(n \log n)$.
The number of boundary spins in the final Hamiltonian therefore scales as $\BigO(n \log n)$.
The final boundary Hamiltonian is a $\left(\Delta, \epsilon, \eta \right)$-simulation of $\Hbulk$.

In order to preserve entanglement wedge reconstruction~\cite{Pastawski:2015}, the location of the spins containing the input state on the Turing machine work tape has to match the location of the original boundary spins.
So, instead of the input tape at the beginning of the $\MPE$ computation containing the input state, followed by a string of $\ket{0}$s, the two are interspersed.
Information about which points on the input tape contain the input state can be included in the description of the Hamiltonian to be simulated.

It is immediate from the definition of the greedy entanglement wedge~\cite[Definition 8]{Pastawski:2015} that bulk local operators in $\mathcal{E}(A)$ can be reconstructed on $A$. The boundary observables/measurements $\{M'\}$ corresponding to bulk observables/measurements $\{M\}$ which have the same outcome, because by definition simulations preserve the outcome of all measurements.
The claim follows.
\end{proof}

It should be noted that the boundary model of the resulting HQECC does not have full rotational invariance.
In order to use the universal Hamiltonian construction the spin chain must have a beginning and end, and the point in the boundary chosen to ``break'' the chain also breaks the rotational invariance.
However, it is possible to construct a HQECC with full rotational symmetry by using a history state Hamiltonian construction with periodic boundary conditions, as in ~\cite[Section 5.8.2]{Gottesman2009}.

In ~\cite[Section 5.8.2]{Gottesman2009} a Turing machine is encoded into a local Hamiltonian acting on a spin chain of length $N$ with periodic boundary conditions.
The ground space of the resulting Hamiltonian is $2N$ fold degenerate.
It consists of history states, where any two adjacent sites along the spin chain can act as boundary spins for the purpose of the Turing machine construction - giving rise to $2N$ distinct ground states.\footnote{The factor of two arises because there is freedom about which of the two adjacent sites is assigned to be the `left' boundary, and which is the `right' boundary.}

We can apply this same idea to construct a rotationally invariant HQECC, which maps a (quasi-)local bulk Hamiltonian, $\Hbulk$ in $\hyper$ to a local Hamiltonian $\Hboundary$ acting on a chain of $N$ qudits.
The code-space of the HQECC is $2N$-fold degnerate, and below the energy cut-off $\Hboundary$ has a direct sum structure:
\begin{equation}
\Hbulk \rightarrow \Hboundary|_{\leq \frac{\Delta}{2}} =
\begin{pmatrix}
\overline{H}_{\mathrm{bulk}} & 0  & \hdots & 0 \\
0 & \overline{H}_{\mathrm{bulk}}  & \hdots & 0 \\
\vdots & \vdots &   \ddots  & 0 \\
0 & 0  & \hdots  & \overline{H}_{\mathrm{bulk}}
\end{pmatrix}
\end{equation}
where each factor in the direct sum acts on one of the possible rotations of the boundary Hilbert space.

Observables are mapped in the same way as the Hamiltonian.
In order to preserve expectation values, we choose the map on states to be of the form:\footnote{See ~\cite[Section 7.1]{Cubitt:2017} for a discussion of maps on states in simulations.}
\begin{equation}
\rho_{\mathrm{boundary}} = \mathcal{E}_{\mathrm{state}}\left(\rho_{\mathrm{bulk}} \right)=
\begin{pmatrix}
\overline{\rho}_{\mathrm{bulk}} & 0  & \hdots & 0 \\
0 & 0 & \hdots & 0 \\
\vdots & \vdots &   \ddots  & 0 \\
0 & 0  & \hdots  & 0
\end{pmatrix}
\end{equation}

We can choose that the bulk state maps into the `unrotated' boundary Hilbert space, so that the geometric relationship between bulk and boundary spins is preserved.\footnote{Although the bulk states maps into one factor of the direct sum structure, every state in the low-energy portion of the boundary does have a bulk interpretation. But most of these states are rotated with respect to the bulk geometry.}

\section{Discussion} \label{discussion}

In this work we have presented a conceptually simple method for proving universality of spin models.
The reliance of this novel method on the ability to encode computation into the low energy subspace of a Hamiltonian suggests that there is a deep connection between universality and complexity.
This insight is made rigorous in \cite{kohler2021general}, where we derive necessary and sufficient conditions for spin systems to be universal simulators (as was done in the classical case \cite{Cubitt:2016}).

This new, simpler proof approach is also stronger, allowing to prove that the simple setting of translationally invariant interactions on a 1D spin chain is sufficient to give universal quantum models.
Furthermore, we have provided the first construction of translationally invariant universal model which is efficient in the number of qudits in the simulator system.

Translationally invariant interactions are more prevalent in condensed matter models than interactions which require fine tuning of individual interaction strengths.
However, a serious impediment to experimentally engineering either of the universal constructions in this paper is the local qudit dimension, which is very large---a problem shared by the earlier 2d translationally invariant construction in~\cite{PiddockBausch}.

An important open question is whether it is possible to reduce the local state dimension in these translationally invariant constructions, while preserving universality.
One possible approach would be to apply the techniques from~\cite{Bausch2016}, which were used to reduce the local dimension of qudits used in translationally invariant \QMA-complete local Hamiltonian constructions.

It would also be interesting to explore what other symmetries universal models can exhibit.
This is of particular interest for constructing HQECC, where we would like the boundary theory to exhibit (a discrete version of) conformal symmetry.

\section*{Acknowledgements}
J.\,B.~acknowledges support from the Draper's Junior Research Fellowship at Pembroke College.
T.\,S.\,C.~is supported by the Royal Society.
T.\,K.~is supported by the EPSRC Centre for Doctoral Training in Delivering Quantum Technologies [EP/L015242/1].
This work was supported by the EPSRC Prosperity Partnership in Quantum Software for Simulation and Modelling (EP/S005021/1).

\printbibliography
\end{document}

